\documentclass[12]{article}
\usepackage{amsmath}

\def \version {September 18, 2013}
\def \dbigcup {\displaystyle\bigcup }
\def \dsum {\displaystyle\sum }
\def \limits { }

\newtheorem{theorem}{Theorem}

\newtheorem{claim}[theorem]{Claim}

\newtheorem{conjecture}[theorem]{Conjecture}

\newtheorem{definition}[theorem]{Definition}

\newtheorem{lemma}[theorem]{Lemma}

\newenvironment{proof}[1][Proof]{\noindent\textbf{#1.} }{\ \rule{0.5em}{0.5em}}

\input{tcilatex}

\begin{document}

\title{Speeding up deciphering by hypergraph ordering}
\author{Peter Horak $^{1,}$\thanks{%
~Supported by a SPIRE grant from University of Bergen, Norway, and by a
research grant from IAS, University of Washington.} \qquad \vspace{2ex}Zsolt
Tuza $^{2,3,}$\thanks{%
~Supported in part by the Hungarian Scientific Research Fund, OTKA grant
T-81493.} \\
{\normalsize $^1$ University of Washington, Tacoma, USA }\\
{\normalsize $^2$ Alfr\'{e}d R\'{e}nyi Institute of Mathematics, Hungarian
Academy of Sciences }\\
{\normalsize $^3$ University of Pannonia, Veszpr\'em, Hungary {\ }}}
\date{{\small Latest update on \version}}
\maketitle

\begin{abstract}
The ``Gluing Algorithm'' of Semaev [Des.\ Codes Cryptogr.\ 49 (2008),
47--60] ---  that finds all solutions of a sparse system of linear equations
over the Galois  field $GF(q)$ --- has average running time  $O(mq^{\max
\left\vert \cup _{1}^{k}X_{j}\right\vert -k}), $  where $m$ is the total
number of equations, and $\cup _{1}^{k}X_{j}$  is the set of all unknowns
actively occurring in the first $k$ equations. Our goal here is to minimize
the exponent of $q$ in the case where every  equation contains at most three
unknowns. 
The main result states that if the total number  $\left\vert \cup
_{1}^{m}X_{j}\right\vert$ of unknowns is equal to $m$,  then the best
achievable exponent is between $c_1m$ and $c_2m$  for some positive
constants $c_1$ and $c_2.$
\end{abstract}

\section{Introduction}

Sparse objects such as sparse matrices, sparse system of (non-)linear
equations occur frequently in science or engineering. For example, huge
sparse matrices often appear when solving partial differential equations. It
seems that \cite{T} \ was the first monograph on the subject, see \cite{P}
for a more a recent one, and \cite{Saad} for a monograph on solving sparse
linear systems of equations.

Nowadays sparse systems are frequently studied in algebraic cryptoanalysis.
First, given a cipher system, one converts it into a system of equations.
Second, the system of equations is solved to retrieve either a key or a
plaintext. As pointed in \cite{C}, this system of equations will be sparse,
since efficient implementations of real-word systems require a low gate
count. Also, as mentioned in \cite{Bardet}, the cryptanalysis of several
modern ciphers reduces to finding the common zeros of $m$ quadratic
polynomials in $n$ unknowns over $F_{2}.$ In the paper \cite{Bardet} an
algorithm reducing the problem to a combination of exhaustive search and
sparse linear algebra in given.

There are plenty of papers on methods for solving a sparse system of
equations. In \cite{Igor2} a so called Gluing Algorithm was designed to
solve such systems over a finite field $GF(q)$. If the set $S_{k}$ of
solutions of the first $k$ equations together with the next equation $%
f_{k+1}=0$ is given then the algorithm constructs the set $S_{k+1}.$ It is
shown there that the average complexity of finding all solutions to the
original system is $O(mq^{\max \left\vert \cup _{1}^{k}X_{j}\right\vert
-k}), $ where $m$ is the total number of equations, and $\cup _{1}^{k}X_{j}$
is the set of all unknowns actively occurring in the first $k$ equations.
Clearly, the complexity of finding all solutions to the system by the Gluing
Algorithm depends on the order of equations. Therefore one is interested to
find a permutation $\pi $ that minimizes the average complexity, and also in
the worst case scenario, i.e., the system of equation for which the average
complexity is maximum. Therefore I. Semaev \cite{Igor} suggested to study
the following combinatorial problem.

Let $\mathcal{S}_{n,m,c}$ be a family of all collections of sets $\mathcal{X=%
}\{X_{1},...,X_{m}\},$ where $X_{i}\subset X,\left\vert X\right\vert =n,$
and $\left\vert X_{i}\right\vert \leq c$ for all $i=1,...,m;$ we allow that
some set may occur in $\mathcal{X}$ more than once. Further, let $\pi $ be a
permutation on $[m]=\{1,...,m\},$ and $1\leq k\leq m.$ Then we set $\Delta (%
\mathcal{X},\pi ,k):=\left\vert \dbigcup\limits_{i=1}^{k}X_{\pi
(i)}\right\vert -k$, and $\Delta (\mathcal{X},\pi ):=\max_{1\leq k\leq
m}\Delta (\mathcal{X},\pi ,k),$ and $\Delta (\mathcal{X}):\mathcal{=}%
\min_{\pi }\Delta (\mathcal{X},\pi ),$ where the minimum runs over all
permutations $\pi $ on $[m].$ Finally, $f_{c}(n,m):=\max_{\mathcal{X}}\Delta
(\mathcal{X}),$ where the maximum is taken over all families $\mathcal{X}$
in $\mathcal{S}_{n,m,c}.$

In this paper we confine ourselves to the case $\left\vert X_{i}\right\vert
\leq 3$ for all $i\in \lbrack m],$ that is, to the case when each equation
of the sparse system contains at most three active variables. We determine $%
f_{2}(n,m)$ for $n\geq 2$ and all $m,$ and also $f_{3}(n,n)$ for $n\leq 9.$
The main result of the paper claims that $f_{3}(n,n)$ grows linearly. More
precisely we show that

\begin{theorem}
\label{bounds} For all $n$ sufficiently large, $f_{3}(n,n)\geq 0.0818757697n$
\mbox{\rm \.=} $\frac{n}{12.2137},$ while for all $n\geq 3,f_{3}(n,n)\leq
\left\lceil \frac{n}{4}\right\rceil +2.$
\end{theorem}

\begin{conjecture}
The quotient $\frac{f_{3}(n,n)}{n}$ tends to a constant as $n\to\infty$.
\end{conjecture}

We point out that after we obtained the above upper bound, an asymptotically
better inequality $f_{3}(n,n)\leq \frac{n}{5}+1+\log _{2}n$ has been proved
in \cite{Igor}. For small $n$ the bound in Theorem \ref{bounds} is slightly
better. 
However, the main reason why we include it in the paper is that it applies
different techniques, and we hope they may have the potential to obtain even
a better bound.

\section{Preliminaries}

In this section we introduce some more needed notions and notation. Several
auxiliary lemmas and observations will be stated as well.\medskip

We start with a lemma that allows one to confine to a special type of
families in $\mathcal{S}_{n,m,c}.$

\begin{lemma}
\label{1}Let $n\geq c,$ there exists a family $\mathcal{X}\in \mathcal{S}%
_{n,m,c}$ so that $\Delta (\mathcal{X})=f_{c}(n,m)$ and $\left\vert
X_{i}\right\vert =c$ for each $i=1,...,m.$
\end{lemma}

\begin{proof}
Let $\mathcal{X=}\{X_{1},...,X_{m}\}$ and $\mathcal{X}^{\prime
}=\{X_{1}^{\prime },...,X_{m}^{\prime }\}$ be in $\mathcal{S}_{n,m,c}$ and $%
X_{i}\subseteq X_{i}^{\prime }$ for all $i.$ Then $\Delta (\mathcal{X)}\leq
\Delta (\mathcal{X}^{\prime })$ and the statement follows.
\end{proof}

\medskip The next observations follow directly from the definition of $%
\Delta (\mathcal{X},\pi \mathcal{)}.\boldsymbol{\ }$

\begin{lemma}
\label{2a} Let $\mathcal{X=\{X}_{1},...,X_{k+s}\},$ $\mathcal{X}_{k}\mathcal{%
=\{X}_{1},...,X_{k}\},\mathcal{Y}_{s}\mathcal{=\{}Y_{1},...,Y_{s}%
\},Y_{i}=X_{k+i}-\dbigcup\limits_{i=1}^{k}X_{i},$ and $\pi _{k}$ and $\pi
^{\prime }$ be the restriction of an ordering $\pi $ of $[k+s]$ to $[k]$ and 
$[k+s]-[k],$ respectively. Then $\newline
$(a) 
\hbox{$\mathbf{\Delta (}\mathcal{X},\pi ,k+s\mathcal{)}=\mathbf{\Delta (}\mathcal{X},\pi ,k\mathcal{)}+\mathbf{\Delta (}\mathcal{Y}_{s},\pi ^{\prime
},s)=\mathbf{\Delta (}\mathcal{X},\pi ,k\mathcal{)+}\left\vert
\dbigcup\limits_{i=k+1}^{s}X_{i}-\dbigcup\limits_{i=1}^{k}X_{i}\right\vert
-s $} $\newline
$(b) $\mathbf{\Delta (}\mathcal{X},\pi )=\max \{\mathbf{\Delta (}\mathcal{X}%
_{k},\pi _{k}\mathcal{)},\mathbf{\Delta (}\mathcal{X},\pi ,k\mathcal{)+}%
\mathbf{\Delta (}\mathcal{Y}_{s},\pi ^{\prime })\}.$
\end{lemma}

Clearly, for each $\mathcal{X\in S}_{n,m,c}$ and all $k\leq m-1,$ we get 
\begin{equation*}
-1\leq \Delta (\mathcal{X},\pi ,k+1)-\Delta (\mathcal{X},\pi ,k)\leq c-1.
\end{equation*}%
The following observation will be frequently used.

\begin{lemma}
\label{2}Let $1\leq s\leq c.$ Then $\Delta (\mathcal{X},\pi ,k+1)-\Delta (%
\mathcal{X},\pi ,k)=s-1$ iff $\ \left\vert X_{\pi
(k+1)}-\dbigcup\limits_{i=1}^{k}X_{\pi (i)}\right\vert =s$.
\end{lemma}

The notions of a connected/disconnected family of sets as well as a
connectivity component will be transferred from the corresponding graph.
More precisely:

\begin{definition}
\label{D}Let $\mathcal{X}=\{X_{1},...,X_{m}\}.$ Then by $G_{\mathcal{X}%
}=(V,E)$ we denote a graph with the vertex set $V=\dbigcup%
\limits_{i=1}^{m}X_{i},$ and $\{i,j\}$ is an edge in $E$ if there is a set $%
X $ in $\mathcal{X}$ so that $\{i,j\}\subset X$. The family $\mathcal{X}$
will be called connected/disconnected if $G_{\mathcal{X}}$ is connected.\ If 
$\mathcal{X}$ is disconnected, and $C=(V_{C},E_{C})$ is a component of $G_{%
\mathcal{X}}$ then the set $V_{C}$ will be called a component of $\mathcal{X}
$. By the order $\left\vert C\right\vert $ of $C$ we mean $\left\vert
V_{C}\right\vert ,$ while by the size $e(C)$ of $C$ we understand the number
of sets $X$ in $\mathcal{X}$ such that $X\subset V_{C}.$
\end{definition}

The following inequality is well known and easy to see.

\begin{lemma}
\label{2b}Let $\mathcal{X}\in S_{n,m,c}$ be connected. Then $m\geq
\left\lceil \frac{n-1}{c-1}\right\rceil .$
\end{lemma}

A standard ordering $\pi $ of sets in $\mathcal{X}$ will be defined
recursively. Choose $X_{\pi (1)}$ in an arbitrary way. After $t\geq 1$ sets
have been ordered (that is, when $\pi (1),...,\pi (t)$ have been set) we
choose $\pi (t+1)$ so that $\left\vert X_{\pi
(t+1)}-\dbigcup\limits_{i=1}^{t}X_{\pi (i)}\right\vert $ is minimum. If $%
\mathcal{X}$ is connected, we have $\left\vert X_{\pi
(t+1)}-\dbigcup\limits_{i=1}^{t}X_{\pi (i)}\right\vert \leq c-1$ for all $%
t\geq 1.$ This in turn implies, see Lemma \ref{2}, that 
\begin{equation}
\text{ for all }t\leq n-1,\text{ }\Delta (\mathcal{X},\pi ,t+1)-\Delta (%
\mathcal{X},\mathcal{\pi },t)\leq c-2  \label{r1}
\end{equation}%
For a disconnected family $\mathcal{X\,\ }$we get that in this case a
standard ordering is obtained by first ordering the components of $\mathcal{X%
}$ and then the sets in the individual components are ordered in a standard
way.

\section{Families with $2$-sets}

In this section we determine the value of $f_{2}(n,m)$ for all $m,n.$ It is
obvious that for a connected family $\mathcal{X\in }S_{n,m,2}$, it is $%
\Delta (X)=1.$ The proof in the case of $\mathcal{X}$ disconnected is more
involved. We note that following key claim is true only for families of $2$%
-sets.

\begin{lemma}
\label{7}Let $\mathcal{X\in }S_{n,m,2}.$ Then there is a standard ordering $%
\pi $ so that $\Delta (\mathcal{X},\pi )=\Delta (\mathcal{X}).$
\end{lemma}

\begin{proof}
Let $\tau $ be an ordering of sets in $\mathcal{X}$ such that $\Delta (%
\mathcal{X},\tau )=\Delta (\mathcal{X}).$ We construct a desired ordering $%
\pi $ in a recursive way. First we set $\pi (1)=\tau (1).$ After $\pi (t)$
has been set (and $t<m$), we define $\pi (t+1)$ as follows. If possible
choose $\pi (t+1)$ such that 
\begin{equation}
\left\vert X_{\pi (t+1)}\cap \dbigcup\limits_{i=1}^{t}X_{\pi (i)}\right\vert
\leq 1  \label{r2}
\end{equation}
is satisfied, otherwise we set $\pi (t+1)=\tau (s),$ where $s$ is the
smallest number such that $X_{\tau (s)}$ has not been ordered yet in the
permutation $\pi .\,$\ It is not difficult to check that for all $\,k\leq m$
we have $\Delta (\mathcal{X},\pi ,k)\leq \Delta (\mathcal{X},\tau ,k).$
\end{proof}

We \ recall that a component of a graph comprising a single vertex is called
a singleton, or trivial.

\begin{theorem}
\label{3} For $n\geq 2$ and all $m,$ $f_{2}(n,m)$ equals the maximum number
of non-trivial components in a simple graph on $n$ vertices with $m$ edges;
i.e., $\ $ $f_{2}(n,m)=m$ for $m\leq \frac{n}{2},$ $f_{2}(n,m)=n-m$ for $%
\frac{n}{2}<m<n-1,$ and $f_{2}(n,m)=1$ for $m\geq n-1.$
\end{theorem}

\begin{proof}
Let $\mathcal{X}=\{X_{1},...,X_{m}\}$ be a family of sets so that $\Delta (%
\mathcal{X)=}f_{2}(n,m).$ By Lemma \ref{1}, we assume that $\left\vert
X_{i}\right\vert =2$ for all $i\in \lbrack n].$ Consider first the case when 
$\mathcal{X}$ is connected; clearly in this case we have $m\geq n-1$. Let $%
\pi $ be a standard ordering of sets in $\mathcal{X}.$ Then $\Delta (%
\mathcal{X},\pi ,1)=1,$ and, by (\ref{r1}), \ $\Delta (\mathcal{X},\pi
,t+1)-\Delta (\mathcal{X},\pi ,t)\leq 0$ for all $k\leq m-1.$ Thus $\Delta (%
\mathcal{X})=f_{2}(n,m)=1.\bigskip $

Suppose now that $\mathcal{X}$ is disconnected. With respect to Lemma \ref{7}%
, we can confine ourselves to standard orderings. As mentioned in
Preliminaries, a standard ordering of a disconnected family $\mathcal{X}$ is
an ordering where first the components of $\mathcal{X}$ are ordered and then
the sets in individual components are ordered in a standard way. For each
component $C$ of $G_{\mathcal{X}}$ we set $d(C):=\left\vert C\right\vert
-e(C).$ Obviously, $d(C)\leq 1$ for each component $C$, and if $%
C_{1},...,C_{t}$ are all components of $G$ then 
\begin{equation}
\dsum\limits_{i=1}^{t}d(C_{i})=\dsum\limits_{i=1}^{t}\left\vert
C_{i}\right\vert -e(C_{i})=\dsum\limits_{i=1}^{t}\left\vert C_{i}\right\vert
-\dsum\limits_{i=1}^{t}e(C_{i})=n-m.  \label{r3}
\end{equation}

\begin{claim}
\label{8}Let $\pi $ be a standard ordering of $\mathcal{X}$. Then 
\begin{equation*}
\Delta (\mathcal{X},\pi )=\max_{0\leq s\leq
t-1}\{1+\dsum\limits_{i=1}^{s}d(C_{\pi (i)})\}.
\end{equation*}
\end{claim}

We have shown above that if $\mathcal{X}$ is connected then $\Delta (%
\mathcal{X})=1.$ To show the statement it suffices to note that, for $k<m,$ $%
\Delta (\mathcal{X},\pi ,k)=1+\dsum\limits_{i=1}^{s}d(C_{\pi (i)}),$ where $%
\ s$ is the number for which $\dsum\limits_{i=1}^{s}e(C_{\pi (i)})\leq
k<\dsum\limits_{i=1}^{s+1}d(C_{\pi (i)}),$ and $\Delta (\mathcal{X},\pi
,m)=n-m.$ From the above claim we immediately get one of key observations:

\begin{claim}
Let $\pi $ be a standard ordering such that the components of $G_{\mathcal{X}%
}$ are ordered in the increasing way with respect to the invariant $d(C),$
and let $\tau $ be any standard ordering of $\mathcal{X}.$ Then $\Delta (%
\mathcal{X},\tau )\geq \Delta (\mathcal{X},\pi )=\Delta (\mathcal{X}).$
\end{claim}

Thus, we can confine ourselves to the ordering $\pi .$ We assume without
loss of generality that $C_{1},...,C_{t}$ is the order of components in this
ordering. Let $m\geq n-1.$ Then, by (\ref{r3}), $\dsum%
\limits_{i=1}^{s}d(C_{i})\leq 0$ for each $s<t$, and, by Lemma \ref{8}, $%
\Delta (\mathcal{X})=1.$ Assume now $m<n-1.$ Then, again by (\ref{r3}) and
Lemma \ref{8}, $\Delta (\mathcal{X},\pi )$ is maximized by a family $%
\mathcal{X}$ with all components $C$ of $\mathcal{X}$ satisfying $d(C)=1,$
thus $\Delta (\mathcal{X},\pi )$ is maximized by a family $\mathcal{X}$
where the corresponding graph $G_{\mathcal{X}\text{ }}$ possesses the
maximum possible number of non-trivial components among all graphs on $n$
vertices and $m$ edges.
\end{proof}

\section{Families with $3$-sets}

For the rest of the paper we deal only with families of $3$-sets. Thus, in $%
f_{3}(n,m)$ we will drop the subscript and write $f(n,m);$ in addition, for
the most interesting case of $n=m,$ we write only $f(n).$

\subsection{Exact values}

There are only a few values of $f(n)$ that we are able to determine
analytically. Here we state only values for $n\leq 9,$ as otherwise
determining the value $f(n)$ is too elaborate as it requires considering a
large number of cases. We start with a rather obvious result that will
simplify the proof of the next theorem.

\begin{lemma}
\label{Claim}For all $n\geq 3,$ $f(n)\leq f(n+1).$
\end{lemma}

\begin{proof}
Let $\mathcal{X=}\{X\,_{1},...,X_{n}\}\mathcal{\in S}_{n,n,3}$ be such that $%
\Delta (\mathcal{X)}=f(n)\,\ $and $z\notin \dbigcup\limits_{i=1}^{n}X_{i}.$
Set $\mathcal{X}^{\prime }=\mathcal{X\cup \{}z\}.$ Let $\pi ^{\prime }$ be
an ordering of sets in $\mathcal{X}^{\prime }.$ Consider the ordering $\pi $
of sets in $\mathcal{X},$ obtained by dropping the set $\{z\}$ from this
order. Then $\Delta (\mathcal{X}^{\prime },\pi ^{\prime })=\Delta (\mathcal{X%
},\pi ),$ and the statement follows.
\end{proof}

\begin{theorem}
$f(3)=2,$ and $f(n)=\left\lceil \frac{n}{3}\right\rceil $ for $4\leq n\leq
9. $
\end{theorem}

\begin{proof}
The statement is obvious for $n=3.$ First we show that, for $4\leq n\leq 9,$ 
$f(n)\geq \left\lceil \frac{n}{3}\right\rceil .$ By $f(3)=2$ and Lemma \ref%
{Claim}, $f(n)\geq 2$ for all $n;$ this proves the lower bound for $4\leq
n\leq 6.$ To see that $f(n)\geq 3$ for $n=7,8,9$ it is sufficient to take
for $\mathcal{X}$ a family such that for any two triples $X,X^{\prime }$ in $%
\mathcal{X}$ it is $\left\vert X\cap X^{\prime }\right\vert \leq 1.$ Then,
for any permutation $\pi $ we get $\Delta (\mathcal{X},\pi ,2)\geq 3,$ that
is, $\Delta (\mathcal{X})\geq 3$. We note that, for $n=7,$ the Fano plane,
and for $n=9,$ any $9$ triples of the unique Steiner triple system $STS(9)$
have the property. For $n=8$, to get the desired family of $8$ triples it
suffices to remove from $STS(9)$ all triples incident with a fixed element $%
x_{0}.$

We note that we are able to prove that $f(n)\leq \left\lceil \frac{n}{3}%
\right\rceil $ for all $n\geq 4.$ This bound is better than the bound $%
f(n)\leq \left\lceil \frac{n}{4}\right\rceil +2,$ proved in this paper, for
a few small values of $n.$ We have not included the proof of $f(n)\leq
\left\lceil \frac{n}{3}\right\rceil $ to this paper as it is quite long. To
have our paper self-contained we prove here $f(n)\leq \left\lceil \frac{n}{3}%
\right\rceil $ only for $n\leq 9.$ In view of Claim \ref{Claim}, it suffices
to show that $f(6)\leq 2,$ and $f(9)\leq 3.$

For $n\in \{6,9\},$ let $\mathcal{X\in S}_{n,n,,3}$ be such that $\Delta (%
\mathcal{X})=f(n),$ and $\left\vert X_{i}\right\vert =3$, $i\in \lbrack n],$
see Lemma \ref{1}$.$ For $\mathcal{X}$ disconnected, the inequality $\Delta (%
\mathcal{X})\leq $ $\left\lceil \frac{n}{3}\right\rceil $ follows from Lemma %
\ref{C}, as the order of the largest component $G_{\mathcal{X}}$ is at most $%
n-3.$ So now we assume that $\mathcal{X}$ is connected. We will construct in
a recursive way an ordering $\pi $ of sets in $\mathcal{X}$ such that $%
\Delta (\mathcal{X},\pi )\leq \left\lceil \frac{n}{3}\right\rceil .$ Let $%
e=\{x,y\}$ be an edge with maximum multiplicity $m(e)=M$ in $G_{\mathcal{X}}$%
. At the beginning of the order $\pi $ come all sets $X_{i}$ with $%
\{x,y\}\subset X_{i}.$ Thus, $\Delta (\mathcal{X},\pi ,k)=2$ for all $k\leq
m $ and $\left\vert \dbigcup\limits_{i=1}^{M}X_{\pi (i)}\right\vert =M+2.$
For $n=6,$ $G_{\mathcal{X}}$ has $18$ edges, thus there is an edge $e$ in $%
G_{\mathcal{X}}$ with multiplicity $m(e)>1.$ Assume that $t<n$ sets in $%
\mathcal{X}$ have been ordered. As $e$ has the maximum multiplicity, for $%
n=6,$ the set $X_{\pi (t+1)}$ can be chosen such that (\ref{r2}) is
satisfied. Thus $\Delta (\mathcal{X},\pi ,k)=2$ for all $k\leq n,$ i.e., $%
\Delta (\mathcal{X})\leq 2$.

So we are left with the case $n=9.$ After $M$ sets containing $x,y$ we
order, in a recursive way, sets, if any, satisfying (\ref{r2}). If we are
able to order in this way all sets of $\mathcal{X}$, then even $\Delta (%
\mathcal{X},\pi ,k)=2$ for all $k\leq n,$ and we are done. Otherwise, as $%
\mathcal{X}$ is connected, we are able to choose as $X_{\pi (t+1)}$ a set
satisfying $\left\vert X_{\pi (t+1)}\cap \dbigcup\limits_{i=1}^{t}X_{\pi
(i)}\right\vert =2.$ Then $\Delta (\mathcal{X},\pi ,k)\leq 3$ for all $k\leq
t+1.$ We note that in all cases, including $M=1,$ we have at this moment $t$
sets ordered with $\left\vert \dbigcup\limits_{i=1}^{t}X_{\pi
(i)}\right\vert \geq 5.$ We leave it to the reader to check that the
remaining sets can be ordered to satisfy (\ref{r2}). The proof is complete.
\end{proof}

\subsection{Lower bound}

\begin{theorem}
For $n$ sufficiently large, we get $f(n)>0.0818757697\,n$.
\end{theorem}

\begin{proof}
We will prove the existence of a family $\mathcal{X=}\{X_{1},...,X_{n}\}\in 
\mathcal{S}_{n.n,3},$ $X_{i}\subset \{x_{1},...,x_{n}\},$ $\ $\ with the
required property $\Delta (\mathcal{X)}\geq 0.0818757697\,n$ using the
following probabilistic model: Select two permutations $\pi $ and $\tau $ on 
$[n]$ randomly and independently; that is, any permutation on $[n]$
coincides with $\pi $ and with $\tau $ with probability $1/n!$, and any
ordered pair of permutations of $[n]$ coincides with $(\pi ,\tau )$ with
probability $(1/n!)^{2}$. \ Set

\begin{equation*}
X_{i}:=\{i,\pi (i),\tau (i)\}\qquad i=1,\dots ,n.
\end{equation*}

We will prove that for $n$ sufficiently large $\mathcal{X}$ \ satisfies $%
\Delta (\mathcal{X)}\geq 0.0818757697\,n$ with a positive probability. Hence
there exists at least one set system having $\Delta (\mathcal{X)}$
sufficiently large. More precisely, we shall prove that there exist positive
constants $c$ and $\varepsilon $ with the following property: The union of
any $cn$ members of $\mathcal{X}$ have cardinality at least $(c+\varepsilon
)n$ with positive probability as $n$ gets large. For simplicity, but without
loss of generality we assume here and also below that $cn$ and $\varepsilon
n $ are integers. This implies that for any ordering $\lambda $ of members
of $\mathcal{X},$where $n$ is sufficiently large, we have $\Delta (\mathcal{X%
},\lambda ,cn)\geq \varepsilon n;$ that is $\Delta (\mathcal{X)}\geq
\varepsilon n.$ Computation will show that the requirement is satisfied if
we put $c=0,4590625$ and $\varepsilon =0.0818757697$. To prove the statement
we will show that $\mathcal{X}$ contains, with the probability strictly less
than $1,$ a subfamily $\{X_{i_{1}},...,X_{i_{m}}\}$ of $k:=cn$ members such
that their union $Y=X_{i_{1}}\cup ...\cup X_{i_{k}}$ is of cardinality at
most $(c+\varepsilon )n.$\medskip

A subfamily of $m$ members can be chosen in $\binom{n}{cn}$ ways$.$ Clearly,
by definition of $X_{i},$ $x_{i_{j}}\in Y$ for $j=1,...,m.$ Therefore, there
are $\binom{n-cn}{\varepsilon n}$ ways how to choose additional $\varepsilon
n$ elements in $Y$. Let $M=\{i_{1},...,i_{m}\}.$ Then $\pi (M)$ can be
chosen in $\binom{(c+\varepsilon )n}{cn}$ ways, and $\pi $ can be defined on 
$M$ in $(cn)!$ ways, while $\pi $ can be defined on $[n]-M$ in $((1-c)n)!$
ways. Since the permutations $\pi $ and $\tau $ have been chosen
independently, the same is valid for $\tau .$ Finally, \ the pair $(\pi
,\tau )$ has been chosen with probability $(n!)^{2}$. Thus, in aggregate,
the probability $p$ that $\mathcal{X}$ contains a subfamily of $cn$ elements
with their union being of cardinality at most $(c+\varepsilon )n$ is 
\begin{equation*}
p\leq \frac{\binom{n}{cn}\binom{n-cn}{\varepsilon n}\binom{(c+\varepsilon )n%
}{cn}^{2}(cn)!^{2}((1-c)n)!^{2}}{(n!)^{2}}
\end{equation*}

which in turn equals 
\begin{equation*}
p\leq \frac{\binom{n}{cn}\binom{n-cn}{\varepsilon n}\binom{(c+\varepsilon )n%
}{cn}^{2}}{\binom{n}{cn}^{2}}=\frac{\binom{(1-c)n}{\varepsilon n}\binom{%
(c+\varepsilon )n}{cn}^{2}}{\binom{n}{cn}}
\end{equation*}

We will calculate $c$ and $\varepsilon $ so that $p<1.$ It is well known
that from Stirling formula we get

\begin{equation*}
\ \ \frac{\log _{2}\binom{x}{ax}}{x}\rightarrow H(a)\text{ \ for }%
x\rightarrow \infty ,\text{ }\ a\text{ fixed, where }H(a)=-a\log
_{2}a-(1-a)\log _{2}(1-a).
\end{equation*}%
$\bigskip $

Thus, taking binary logarithm of $p<1$ we get that the inequality holds for
every sufficiently large $n$ if 
\begin{equation*}
\log _{2}\binom{(1-c)n}{\varepsilon n}+2\log _{2}\binom{(c+\varepsilon )n}{cn%
}-\log _{2}\binom{n}{cn}<0,
\end{equation*}

that is, if 
\begin{equation*}
\frac{1}{n}\left[\frac{1-c}{1-c}\log _{2}\binom{(1-c)n}{\frac{\varepsilon }{%
1-c}(1-c)n}+\frac{c+\varepsilon }{c+\varepsilon }2\log _{2}\binom{%
(c+\varepsilon )n}{\frac{c}{c+\varepsilon }(c+\varepsilon )n}-\log _{2}%
\binom{n}{cn}\right]<0,
\end{equation*}

hence%
\begin{equation*}
(1-c)H(\frac{\varepsilon }{1-c})+2(c+\varepsilon )H(\frac{c}{c+\varepsilon }%
)-H(c)<0.
\end{equation*}

Substituting the values $c=0.4590625$ and $\varepsilon =0.0818757697241$ one
can check that the left side of the above inequality is approximately $%
-0.0000000000005$ and hence strictly negative, consequently the probability $%
p<1$ for $n$ large enough. The proof is complete.
\end{proof}

\subsection{Upper bound}

In this section we will prove that, for all $n\geq 4,f(n)\leq \left\lceil 
\frac{n}{4}\right\rceil +2.$ We will start with a series of auxiliary upper
bounds. The first one looks to be fairly crude but for $m$ small with
respect to $n$ it is sharp.

\begin{lemma}
\label{4} For all $n,m$ we get $f(n,m)\leq 2\left\lceil \frac{n}{3}%
\right\rceil $.
\end{lemma}

\begin{proof}
Let $\mathcal{X\in S}_{n,m,3}$ and $\pi $ be a permutation on $[m].$ For any 
$k\leq \left\lceil \frac{n}{3}\right\rceil ,$ $\Delta (\mathcal{X},\pi
,k)=\left\vert \dbigcup\limits_{i=1}^{k}X_{\pi (i)}\right\vert -k\leq
3k-k\leq 2\left\lceil \frac{n}{3}\right\rceil .$ Otherwise, $\left\vert
\dbigcup\limits_{i=1}^{k}X_{\pi (i)}\right\vert -k\leq n-k\leq
n-(\left\lceil \frac{n}{3}\right\rceil +1)\leq 2\left\lceil \frac{n}{3}%
\right\rceil .$
\end{proof}

\medskip A better bound can be obtained if $m$ is sufficiently large. Also
in this case the bound is sharp for some values of $m.$

\begin{lemma}
\label{5}For $m\geq \left\lceil \frac{n-1}{2}\right\rceil ,$ we have $%
f(n,m)\leq \left\lfloor \frac{n+1}{2}\right\rfloor .$
\end{lemma}

\begin{proof}
Let $\mathcal{X}=\{X_{1},...,X_{m}\}$ be a family of sets such that $\Delta (%
\mathcal{X)=}f(n,m).$ By Lemma \ref{1}, we assume that $\left\vert
X_{i}\right\vert =3$ for all $i\in \lbrack m].$ Consider first a case when $%
\mathcal{X}$ is connected, and let $\pi $ be a standard ordering of $%
\mathcal{X}$. Then $\Delta (\mathcal{X},\pi ,1)=2,$ and by (\ref{r1}), $%
\Delta (\mathcal{X},\pi ,t+1)-\Delta (\mathcal{X},\mathcal{\pi },t)\leq 1,$
for all $t\geq 1.$ As $\Delta (\mathcal{X},\pi ,t+1)-\Delta (\mathcal{X},%
\mathcal{\pi },t)=1$ implies $\left\vert X_{\pi
(t+1)}-\dbigcup\limits_{i=1}^{t}X_{\pi (i)}\right\vert =2,$ we have that $%
\Delta (\mathcal{X},\pi )\leq 2+\left\lfloor \frac{n-3}{2}\right\rfloor
=\left\lfloor \frac{n+1}{2}\right\rfloor ,$ thus $\Delta (\mathcal{X})\leq
\left\lfloor \frac{n+1}{2}\right\rfloor .$\medskip

Now let $\mathcal{X}$ be disconnected and $C_{1},...,C_{s}$ be components of
the graph $G_{\mathcal{X}}.$ By Lemma \ref{2b}, a connected family $\mathcal{%
Y\in S}_{n,m,3}$ has to contain at least $\left\lceil \frac{n-1}{2}%
\right\rceil $ triples. We define $\gamma (C_{i})=e(C_{i})-\left\lceil \frac{%
\left\vert C_{i}\right\vert -1}{2}\right\rceil .$ Thus, $\gamma (C_{i})\geq
0 $ for all $i\in \lbrack s].$ Moreover, $\dsum\limits_{i=1}^{s}\gamma
(C_{i})= $ $\dsum\limits_{i=1}^{s}(e(C_{i})-\left\lceil \frac{\left\vert
C_{i}\right\vert -1}{2}\right\rceil )=m-\dsum\limits_{\left\vert
C_{i}\right\vert \text{ odd}}\left\lceil \frac{\left\vert C_{i}\right\vert -1%
}{2}\right\rceil -\dsum\limits_{\left\vert C_{i}\right\vert \text{ even}%
}\left\lceil \frac{\left\vert C_{i}\right\vert -1}{2}\right\rceil =$ $%
m-\dsum\limits_{\left\vert C_{i}\right\vert \text{ odd}}\frac{\left\vert
C_{i}\right\vert -1}{2}-\dsum\limits_{\left\vert C_{i}\right\vert \text{ even%
}}\frac{\left\vert C_{i}\right\vert }{2}=m-\frac{n}{2}+\frac{odd}{2}\geq
\left\lceil \frac{n-1}{2}\right\rceil -\frac{n}{2}+\frac{odd}{2}%
=\left\lfloor \frac{odd}{2}\right\rfloor ,$ where $odd$ is the number of the
odd order components in $G_{\mathcal{X}}.$ Let $\pi $ be a standard ordering
on $\mathcal{X}$ where the components $C_{i}$ are ordered in the decreasing
manner with respect to $\gamma ;$ without loss of generality we assume that $%
C_{1},...,C_{s}$ is this ordering. Then $\dsum\limits_{i=1}^{s}\gamma
(C_{i})\geq \left\lfloor \frac{odd}{2}\right\rfloor $ implies that, for all $%
t\leq s,$

\begin{equation}
\dsum\limits_{i=1}^{t}\gamma (C_{i})\geq \min \{t,\left\lfloor \frac{odd}{2}%
\right\rfloor \}.  \label{odd}
\end{equation}

Now we show that for every $k,1\leq k\leq m,$ we have $\Delta (\mathcal{X}%
,\pi ,k)\leq \left\lfloor \frac{n+1}{2}\right\rfloor .$ For $k<m,$ there is
a unique $t$ such that $\dsum\limits_{i=1}^{t}e(C_{i})\leq
k<\dsum\limits_{i=1}^{t+1}e(C_{i}).$ By Lemma \ref{2a}(a) we get%
\begin{equation*}
\Delta (\mathcal{X},\pi ,k)=\Delta (\mathcal{X},\pi
,\dsum\limits_{i=1}^{t}e(C_{i}))+\Delta (\mathcal{X}_{i},\pi
_{i},k-\dsum\limits_{i=1}^{t}e(C_{i})),
\end{equation*}

where $\mathcal{X}_{i}$ is a subfamily of $\mathcal{X}$ comprising triples
that are subsets of $C_{i},$ and $\pi _{i}$ is the restriction of $\pi $ to $%
\mathcal{X}_{i}$. From the case of a connected family $\mathcal{X}$
discussed above we have $\Delta (\mathcal{X}_{i},\pi
_{i},k-\dsum\limits_{i=1}^{t}e(C_{i}))\leq \left\lfloor \frac{\left\vert
C_{i}\right\vert +1}{2}\right\rfloor .$ Denote by $odd_{t}$ the number of
odd components among $C_{1},...,C_{t}.$ Then, $\Delta (\mathcal{X}\mathbf{,}%
\pi ,k\mathbf{)}\mathbf{\leq }\dsum\limits_{i=1}^{t}\left\vert
C_{i}\right\vert -\dsum\limits_{i=1}^{t}e(C_{i})+\left\lfloor \frac{%
\left\vert C_{t+1}\right\vert +1}{2}\right\rfloor \leq \left\lfloor \frac{%
\left\vert C_{t+1}\right\vert +1}{2}\right\rfloor
+\dsum\limits_{i=1}^{t}(\left\lfloor \frac{\left\vert C_{i}\right\vert +1}{2}%
\right\rfloor -\gamma (C_{i}))\leq $

$\left\lfloor \frac{\left\vert C_{t+1}\right\vert +1}{2}\right\rfloor
+\dsum\limits_{\left\vert C_{i}\right\vert \text{even}}\frac{\left\vert
C_{i}\right\vert }{2}+\dsum\limits_{\left\vert C_{i}\right\vert \text{odd}}%
\frac{\left\vert C_{i}\right\vert }{2}+\frac{odd_{t}}{2}-\dsum%
\limits_{i=1}^{t}\gamma (C_{i})\leq \left\lfloor \frac{n+1}{2}\right\rfloor $
as,

by (\ref{odd}), $\frac{odd_{t}}{2}\leq \dsum\limits_{i=1}^{t}\gamma (C_{i}),$
and$\dsum\limits_{\left\vert C_{i}\right\vert \text{even}}\frac{\left\vert
C_{i}\right\vert }{2}+\dsum\limits_{\left\vert C_{i}\right\vert \text{odd}}%
\frac{\left\vert C_{i}\right\vert }{2}+\left\lfloor \frac{\left\vert
C_{t+1}\right\vert +1}{2}\right\rfloor \leq \left\lfloor \frac{n+1}{2}%
\right\rfloor $.
\end{proof}

\medskip The next auxiliary bound deals with the case when $\mathcal{X}$ is
disconnected.

\begin{lemma}
\label{C}If $\mathcal{X\in S}_{n,m,3}$ is disconnected, and $m\geq n$, then $%
\Delta (\mathcal{X})\leq \left\lfloor \frac{c+1}{2}\right\rfloor ,$ where $c$
is the order of the largest component of $G_{\mathcal{X}}$.
\end{lemma}

\begin{proof}
Let $C_{1},...,C_{s}$ be components of $G_{\mathcal{X}}\,,$ and let $%
\mathcal{X}_{i}$ be the subfamily of $\mathcal{X}$ comprising triples that
are subsists of $C_{i}.$ As in the proof of Theorem \ref{3}, \ we set $%
d(C)=\left\vert C_{i}\right\vert -e(C_{i})$ and get $\dsum%
\limits_{i=1}^{s}d(C_{i})=n-m\leq 0.$ Consider the standard ordering $\pi $
of $\mathcal{X}$ such that the components are ordered in the increasing way
with respect to the invariant $d;$ without loss of generality. we assume
that $C_{1},...,C_{s}$ is such ordering. Then $\dsum%
\limits_{i=1}^{t}d(C_{i})\leq 0$ for any $t\leq s.$ Let $\pi _{i}$ be a
restriction of $\pi $ to $\mathcal{X}_{i}.$ Then, by Lemma \ref{5}, for each
component $C_{i},$ we have

\begin{equation*}
\max_{1\leq k\leq e(C_{i})}\Delta (\mathcal{X}_{i},\pi _{i},k)\leq
\left\lfloor \frac{\left\vert C_{i}\right\vert +1}{2}\right\rfloor
\end{equation*}%
Further, $\ \Delta (\mathcal{X}_{i},\pi _{i},e(C_{i}))=\left\vert
C_{i}\right\vert -e(C_{i})=d(C_{i}).$ Extending this conclusion to $\pi $ we
have:$\mathcal{\ }$If $a$ is the total number of triples in the first $t$
components, then

\begin{equation*}
\Delta (\mathcal{X},\pi ,a)=\dsum\limits_{i=1}^{t}d(C_{i})\leq 0.
\end{equation*}

Let $k<m.$ Then there is a uniquely determined number $t$ such that $%
\dsum\limits_{i=1}^{t-1}e(C_{i})\leq k<\dsum\limits_{i=1}^{t}e(C_{i}).$ Set $%
a=\dsum\limits_{i=1}^{t-1}e(C_{i}).$ By Lemma \ref{2a}(a), $\Delta (\mathcal{%
X},\pi ,k)=\Delta (\mathcal{X},\pi ,a)+\Delta (\mathcal{X}_{t},\pi
_{t},k-a)\leq \Delta (\mathcal{X}_{t},\pi _{t},k-a)\leq \left\lfloor \frac{%
\left\vert C_{i}\right\vert +1}{2}\right\rfloor $. The proof is complete.
\end{proof}

Before proving the upper bound we state one more lemma.

\begin{lemma}
\label{6}Let $\mathcal{X\in S}_{n,n,3},$ and let $\Delta (\mathcal{X},\pi
,k)=\left\lceil \frac{n}{4}\right\rceil +1.$ Then there is an $\varepsilon
\geq 0$ such that $k=\left\lceil \frac{n}{4}\right\rceil +\varepsilon ,$ and 
$\left\vert \dbigcup\limits_{i=1}^{k}X_{\pi (i)}\right\vert =\left\lceil 
\frac{n}{2}\right\rceil +\varepsilon .$
\end{lemma}

\begin{proof}
First let $\mathcal{X}$ be connected, and $\tau $ be a standard ordering of $%
\mathcal{X}.$ By (\ref{r1}), for $s\geq 1$ we have $\Delta (\mathcal{X},\tau
,s+1)-\Delta (\mathcal{X},\mathcal{\tau },s)\leq 1.$ Therefore $\Delta (%
\mathcal{X},\tau ,s)\leq 2+(s-1)=s+1.$ Thus, if $\Delta (\mathcal{X},\tau
,s)=\left\lceil \frac{n}{4}\right\rceil +1$ then $s\geq \left\lceil \frac{n}{%
4}\right\rceil .$ Then, by the definition of $\pi $ and $k$ it is $k\geq
\left\lceil \frac{n}{4}\right\rceil ,$ i.e., $k=\left\lceil \frac{n}{4}%
\right\rceil +\varepsilon \,$\ for some $\varepsilon \geq 0,$ which in turn
implies $\left\vert M\right\vert =k+\Delta (\mathcal{X},\pi ,k)\geq
\left\lceil \frac{n}{2}\right\rceil +\varepsilon .\bigskip $

Now let $\mathcal{X}$ be disconnected. By Lemma \ref{C}, $\Delta (\mathcal{%
X)\leq }\left\lfloor \frac{\left\vert C\right\vert +1}{2}\right\rfloor ,$
where $C$ is the largest component of $G_{\mathcal{X}}.$ As $\Delta (%
\mathcal{X)>}\left\lceil \frac{n}{4}\right\rceil +2,$ we get $\left\vert
C\right\vert \geq \left\lceil \frac{n}{2}\right\rceil +2.$ Since the
subfamily of $\mathcal{X}$ with its triples in $C$ is connected, $e(C)\geq
\left\lceil \frac{\left\vert C\right\vert -1}{2}\right\rceil \geq
\left\lceil \frac{n}{4}\right\rceil .$ Now it suffices to repeat the
argument used in the case $\mathcal{X}$ is connected.
\end{proof}

\begin{theorem}
For all $n\geq 4,$ $f(n)\leq \left\lceil \frac{n}{4}\right\rceil +2.$
\end{theorem}

\begin{proof}
Let $\mathcal{X}=\{X_{1},...,X_{n}\}$ \ be such that $\Delta (\mathcal{X}%
)=f(n)$ and, see Lemma \ref{1}, $\left\vert X_{i}\right\vert =3$ for all $%
i\in \lbrack n].$ Assume by contradiction that $\Delta (\mathcal{X}%
)>\left\lceil \frac{n}{4}\right\rceil +2.$ We choose an ordering $\pi $ of $%
\mathcal{X}$ and a number $k\in \lbrack n]$ so that $k$ is the largest
number \ with the property (a) $\Delta (\mathcal{X},\pi ,k)=\left\lceil 
\frac{n}{4}\right\rceil +1,$ and $\Delta (\mathcal{X},\pi ,s)\leq
\left\lceil \frac{n}{4}\right\rceil +1$ for all $s\leq k-1,$ (b) for all
orderings $\tau $ of $\mathcal{X}$ there is $s_{\tau }\leq k+1$ such that $%
\Delta (\mathcal{X},\tau ,s_{\tau })>\left\lceil \frac{n}{4}\right\rceil +1.$
From (a ) we have that $\Delta (\mathcal{X},\pi ,k+1)>\left\lceil \frac{n}{4}%
\right\rceil $ $+1,$ and by Lemma \ref{2}, $\ $%
\begin{equation*}
\left\vert X_{\pi (t)}-\dbigcup\limits_{i=1}^{k}X_{\pi (i)}\right\vert \geq 2%
\text{ for all }t>k.
\end{equation*}%
Set $M=\dbigcup\limits_{i=1}^{k}X_{\pi (i)}.$ Let $\overline{M}$ denote the
complement to $M$ with respect to the underlying set. By Lemma \ref{3}, $%
k\geq \left\lceil \frac{n}{4}\right\rceil +\varepsilon ,$ and $\left\vert 
\overline{M}\right\vert =n-(\left\lceil \frac{n}{2}\right\rceil +\varepsilon
)=\left\lfloor \frac{n}{2}\right\rfloor -\varepsilon ,$ where $\varepsilon
\geq 0.$ Further, for $i=2,3,$ let $\mathcal{A}_{i},$ $a_{i}=\left\vert 
\mathcal{A}_{i}\right\vert $, be the subfamily of $\mathcal{X}$, such that $%
X\in \mathcal{A}_{i}$ if $\left\vert X\cap \overline{M}\right\vert =i.$ We
note that $\mathcal{A}_{2}\cup \mathcal{A}_{3}$ comprises $n-k$ sets of $%
\mathcal{X}$ that come in the ordering $\pi $ after $X_{\pi (k)}.$ We choose 
$\pi $ so that all sets in $\mathcal{A}_{2}$ come in the ordering $\pi $
before sets from $\mathcal{A}_{3}.$ Two cases are considered.$\bigskip $

First, let $a_{2}=\left\vert \overline{M}\right\vert +\alpha ,$ where $%
\alpha \geq 0.$ Let $\mathcal{B}_{2}$ be a family of $2$-sets, $\mathcal{B}%
_{2}=\{$ $B;\,$ $B=\overline{M}\cap X$ for some $X\in \mathcal{A}_{2}\}.$
Then, by Lemma \ref{2a}(a), for each $t,$ $k$ $+1\leq t\leq k+a_{2},$ we
have $\Delta (\mathcal{X},\pi ,t)=\Delta (\mathcal{X},\pi ,k)+\Delta (%
\mathcal{B}_{2},\pi ^{\prime },t-k)=\left\lceil \frac{n}{4}\right\rceil
+1+\Delta (\mathcal{B}_{2},\pi ^{\prime },t-k),$ where $\pi ^{\prime }$ is
the restriction of $\pi $ to the set $\{k+1,...,k+a_{2}\}.$ By Theorem \ref%
{3}, for $m\geq n,$ $f_{2}(n,m)=1.$ Hence, if we choose $\pi ^{\prime }$ to
be the same permutation as in the proof of Theorem \ref{3}, then $\Delta (%
\mathcal{B}_{2},\pi ^{\prime },t-k)\leq 1$ for all $1\leq t-k\leq a_{2}.$
Thus, $\Delta (\mathcal{X},\pi ,s)\leq \left\lceil \frac{n}{4}\right\rceil
+2 $ for all $k+1\leq s\leq k+a_{2}.$ Let $B_{2}=\dbigcup\limits_{B\in 
\mathcal{B}_{2}}B.$ By Lemma \ref{2a}(a), $\Delta (\mathcal{X},\pi
,k+a_{2})=\Delta (\mathcal{X},\pi ,k)+\Delta (\mathcal{B}_{2},\pi ^{\prime
},a_{2})=\left\lceil \frac{n}{4}\right\rceil +1+(\left\vert B_{2}\right\vert
-a_{2}).$ To finish the proof of this part we will show that $\Delta (X,\pi
,s)\leq \left\lceil \frac{n}{4}\right\rceil +2$ is true also for all $s,$ $%
k+a_{2}+1\leq s\leq n.$ Again by Lemma \ref{2a}(a), 
\begin{equation*}
\Delta (\mathcal{X},\pi ,s)=\Delta (\mathcal{X},\pi ,k+a_{2})+\Delta (%
\mathcal{A}_{3}^{\ast },\pi ^{\ast },s-(k+a_{2}))
\end{equation*}%
where $\mathcal{A}_{3}^{\ast }$ consists of sets $X\cap (\overline{M}%
-B_{2}), $ $X\in \mathcal{A}_{3},$ and $\pi ^{\ast }$ is the restriction of $%
\pi $ to $\mathcal{A}_{3}^{\ast }.$ By Lemma \ref{4}, for any $%
s-(k+a_{2})\leq a_{3}$ we have $\Delta (\mathcal{A}_{3}^{\ast },\pi ^{\ast
},s-(k+a_{2}))\leq \frac{2}{3}(\left\vert \overline{M}\right\vert
-\left\vert B_{2}\right\vert ).$ Hence $\Delta (\mathcal{X},\pi
,s)=\left\lceil \frac{n}{4}\right\rceil +1+(\left\vert B_{2}\right\vert
-\left\vert \overline{M}\right\vert -\alpha )+\frac{2}{3}(\left\vert 
\overline{M}\right\vert -\left\vert B_{2}\right\vert )\leq \left\lceil \frac{%
n}{4}\right\rceil +1$ as $B_{2}\subset \overline{M},$ i.e., $\left\vert
B_{2}\right\vert \leq \left\vert \overline{M}\right\vert $ and $\alpha \geq
0.$\medskip

We are left with the case $a_{2}=\left\vert \overline{M}\right\vert -\alpha
, $ where $\alpha >0.$ We consider an ordering $\tau ,$ where the triples in 
$\mathcal{A}_{3}$ come at the very beginning of this ordering, followed by
triples from $\mathcal{A}_{2}$. At the very end of the ordering come triples
in $\mathcal{X-(A}_{2}\cup \mathcal{A}_{3}\mathcal{)}$ in the same order as
in the ordering $\pi .$ As $\left\vert \overline{M}\right\vert =\left\lfloor 
\frac{n}{2}\right\rfloor -\varepsilon ,$ and $\alpha >0,$ we have $%
a_{3}=n-k-a_{2}=n-(\left\lceil \frac{n}{4}\right\rceil +\varepsilon
)-(\left\vert \overline{M}\right\vert -\alpha )\geq \left\lfloor \frac{3n}{4}%
\right\rfloor -\varepsilon -\left\lfloor \frac{n}{2}\right\rfloor
+\varepsilon -\alpha \geq \left\lfloor \frac{n}{4}\right\rfloor \geq
\left\lceil \frac{\left\vert \overline{M}\right\vert -1}{2}\right\rceil .$
Therefore, by Lemma \ref{5}, for all $s\leq a_{3},$ we have $\Delta (%
\mathcal{X},\tau ,s)\leq \left\lfloor \frac{\left\vert \overline{M}%
\right\vert +1}{2}\right\rfloor \leq \left\lceil \frac{n}{4}\right\rceil +1.$
Let $B_{3}=:\dbigcup\limits_{X\in \mathcal{A}_{3}}X.$ Then, $\Delta (%
\mathcal{X},\tau ,a_{3})=\left\vert B_{3}\right\vert -a_{3}.$ We get, by
Lemma \ref{2a}(b), 
\begin{equation*}
\max_{1\leq s\leq a_{2}+a_{3}}\Delta (\mathcal{X},\tau ,s)=\max
\{\max_{1\leq s\leq a_{3}}\Delta (\mathcal{X},\tau ,s),\max_{a_{3}+1\leq
s\leq a_{2}+a_{3}}\Delta (\mathcal{X},\tau ,s)\}\leq
\end{equation*}

\begin{equation*}
\max \{\left\lceil \frac{n}{4}\right\rceil +1,\Delta (\mathcal{X},\tau
,a_{3})+\max_{1\leq s\leq a_{2}}\Delta (\mathcal{A}_{2}^{\ast },\tau
^{\prime },s)\},
\end{equation*}%
where $\tau ^{\prime }$ is the restriction of $\tau $ to $%
\{a_{3}+1,...,a_{3}+a_{2}\}$ and $\mathcal{A}_{2}^{\ast }$ comprises sets $%
X-B_{3},X\in \mathcal{A}_{2}$. As $\left\vert X\cap M\right\vert =1$ for all 
$X\in \mathcal{A}_{2}$ we further get 
\begin{eqnarray*}
\max_{1\leq s\leq a_{2}+a_{3}}\Delta (\mathcal{X},\tau ,s) &\leq &\max
\{\left\lceil \frac{n}{4}\right\rceil +1,\left\vert B_{3}\right\vert
-a_{3}+\max_{1\leq s\leq a_{2}}\left\vert \dbigcup\limits_{i=1}^{s}X_{\tau
(i)}\cap (\overline{M}-B_{3})\right\vert \} \\
&\leq &\max \{\left\lceil \frac{n}{4}\right\rceil +1,\left\vert
B_{3}\right\vert -a_{3}+\left\vert \overline{M}-B_{3}\right\vert \}\leq \max
\{\left\lceil \frac{n}{4}\right\rceil +1,\left\vert \overline{M}\right\vert
-a_{3}\},
\end{eqnarray*}

since $B_{3}\subset \overline{M}$. Finally, because $a_{3}\geq \left\lceil 
\frac{\left\vert \overline{M}\right\vert -1}{2}\right\rceil $ and $%
\left\vert \overline{M}\right\vert =\left\lfloor \frac{n}{2}\right\rfloor
-\varepsilon ,$ we get 
\begin{equation*}
\max_{1\leq s\leq a_{2}+a_{3}}\Delta (\mathcal{X},\tau ,s)\leq \left\lceil 
\frac{n}{4}\right\rceil +1
\end{equation*}%
$\bigskip $

Therefore, by the part (b) of definition of the value of $k$ and the
permutation $\pi ,$ we have $a_{2}+a_{3}\leq k$. Since $k=n-a_{2}-a_{3},$ we
get $k\geq \left\lceil \frac{n}{2}\right\rceil .$ Hence $k=\left\lceil \frac{%
n}{2}\right\rceil +\varepsilon ^{\prime }$ for some $\varepsilon ^{\prime
}\geq 0,$ and $\left\vert M\right\vert =\left\vert
\dbigcup\limits_{i=1}^{k}X_{\pi (i)}\right\vert =k+\Delta (\mathcal{X},\pi
,k)=\left\lceil \frac{n}{2}\right\rceil +\varepsilon ^{\prime }+\left\lceil 
\frac{n}{4}\right\rceil +1\geq \left\lceil \frac{3n}{4}\right\rceil
+\varepsilon ^{\prime }.$

We have $a_{2}=\left\vert \overline{M}\right\vert -\alpha ,\alpha >0.$
Therefore,%
\begin{equation*}
a_{3}>\left\vert \overline{M}\right\vert .
\end{equation*}%
Indeed, $a_{3}=n-k-a_{2}\geq n-\left\lceil \frac{n}{2}\right\rceil
-\varepsilon ^{\prime }-\left\vert \overline{M}\right\vert +\alpha \geq
\left\lfloor \frac{n}{2}\right\rfloor -\varepsilon ^{\prime }-\left\lfloor 
\frac{n}{4}\right\rfloor +\varepsilon ^{\prime }+\alpha \geq \left\lceil 
\frac{n}{4}\right\rceil >\left\vert \overline{M}\right\vert .$

Thus, as $a_{3}\geq \left\vert \overline{M}\right\vert \geq \left\vert
B_{3}\right\vert ,$ by Lemma \ref{5} , for each $s\leq a_{3,}$ $\Delta (%
\mathcal{X},\tau ,s)\leq \left\lfloor \frac{\left\vert \overline{M}%
\right\vert +1}{2}\right\rfloor \leq \left\lceil \frac{n}{4}\right\rceil +1.$
We note that $\Delta (\mathcal{X},\tau ,a_{3})=\left\vert B_{3}\right\vert
-a_{3}<0.$ Further, by Lemma \ref{2a}(b), we get

\begin{equation*}
\max_{1\leq s\leq a_{2}+a_{3}}\Delta (\mathcal{X},\tau ,s)=\max
\{\max_{1\leq s\leq a_{3}}\Delta (\mathcal{X},\tau ,s),\max_{a_{3}+1\leq
s\leq a_{2}+a_{3}}\Delta (\mathcal{X},\tau ,s)\}\leq
\end{equation*}

\begin{equation*}
\max \{\left\lceil \frac{n}{4}\right\rceil +1,\Delta (\mathcal{X},\tau
,a_{3})+\max_{1\leq s\leq a_{2}}\Delta (\mathcal{A}_{2}^{\ast },\tau
^{\prime },s)\}\leq
\end{equation*}%
\begin{equation*}
\max \{\left\lceil \frac{n}{4}\right\rceil +1,0+\max_{1\leq s\leq
a_{2}}\left\vert \dbigcup\limits_{i=1}^{s}X_{\tau (i)}\cap (\overline{M}%
-B_{3})\right\vert \}\leq
\end{equation*}

\begin{equation*}
\max \{\left\lceil \frac{n}{4}\right\rceil +1,\left\vert \overline{M}%
-B_{3}\right\vert \}\leq \left\lceil \frac{n}{4}\right\rceil +1
\end{equation*}
because $a_{3}\geq \left\vert \overline{M}\right\vert $ and $B_{3}\subset 
\overline{M}.$

Clearly, $\Delta (\mathcal{X},\tau ,a_{2}+a_{3})=\left\vert \overline{M}%
\right\vert +a_{2}-(a_{2}+a_{3})<0.$ For $t=s+a_{2}+a_{3},s\leq k,$ we get
by Lemma \ref{2a}(a), 
\begin{eqnarray*}
\Delta (\mathcal{X},\tau ,t) &=&\Delta (\mathcal{X},\tau
,a_{2}+a_{3})+\left\vert \dbigcup\limits_{i=a_{2}+a_{3}+1}^{t}X_{\tau
(i)}-\dbigcup\limits_{i=1}^{a_{2}+a_{3}}X_{\tau (i)}\right\vert -s\,< \\
\left\vert \dbigcup\limits_{i=a_{2}+a_{3}+1}^{t}X_{\tau (i)}\cap
M\right\vert -s &=&\left\vert \dbigcup\limits_{i=1}^{s}X_{\pi
(i)}\right\vert -s=\Delta (\mathcal{X},\pi ,s)\leq \left\lceil \frac{n}{4}%
\right\rceil +1.
\end{eqnarray*}

We recall that triples not in $\mathcal{A}_{2}\cup \mathcal{A}_{3}$ are in $%
\tau $ ordered the same way as in $\pi $ and $X_{\pi (s)}\subset M$ for all $%
s\leq k.$ We proved that $\Delta (\mathcal{X},\pi ,s)\leq \left\lceil \frac{n%
}{4}\right\rceil +1$ for all $1\leq s\leq n,$ which contradicts that $\Delta
(\mathcal{X})>\left\lceil \frac{n}{4}\right\rceil +2.$ The proof is complete.
\end{proof}

\paragraph{Acknowledgements.}

The authors are indebted to Noga Alon for discussions on expanders and on 
probabilistic methods, which lead to an improvement of the lower bound. 
This research of the  first author was initiated thanks to the University of
Washington  -- University of Bergen exchange program.

\end{document}